\documentclass[conference,letterpaper]{IEEEtran}

\usepackage{enumerate}
\usepackage[shortlabels]{enumitem}

\usepackage{amssymb,amsfonts}
\usepackage{mathrsfs}
\usepackage{amsthm}
\usepackage[cmex10]{amsmath} 

\usepackage{algorithm}
\usepackage[noend]{algpseudocode}

\usepackage{multicol}
\algrenewcommand\algorithmiccomment[1]{\textcolor{lightgray}{\hfill // #1}}

\usepackage{graphicx}
\usepackage{xcolor}
\usepackage{changes}
\usepackage{ulem}
\usepackage{bbm}
\usepackage{lipsum,adjustbox}

\usepackage{stmaryrd}

\usepackage{hyperref}
\hypersetup{
 colorlinks,
 linkcolor={blue!100!black},
 citecolor={blue!100!black},
 urlcolor={blue!80!black},
}

\usepackage[style=ieee,
            backend=biber,
            doi=false,
            ]{biblatex}

\usepackage{subcaption}
\usepackage{wrapfig}

\theoremstyle{definition}
\newtheorem{theorem}{Theorem}

\newtheorem{example}{Example}
\newtheorem{lemma}{Lemma}

\newtheorem{corollary}{Corollary}

\newtheorem{definition}{Definition}
\newtheorem{remark}{Remark}

\usepackage{soul} %For strikesout, /st

\newcommand{\ceil}[1]{ {\lceil#1\rceil}}
\newcommand{\floor}[1]{ {\lfloor#1\rfloor}}
\newcommand{\bi}{{\{0,1\}}}

\newcommand{\hw}{\operatorname{wt}}
\newcommand{\abs}[1]{\vert #1 \vert}

\newcommand{\bbF}{\mathbb{F}}

\newcommand{\bfc}{\mathbf{c}}

\newcommand{\bff}{\mathbf{f}}

\newcommand{\bfm}{\mathbf{m}}

\newcommand{\bfr}{\mathbf{r}}

\newcommand{\bfx}{\mathbf{x}}
\newcommand{\bfy}{\mathbf{y}}
\newcommand{\bfz}{\mathbf{z}}

\newcommand{\cC}{\mathcal{C}}
\newcommand{\cD}{\mathcal{D}}

\newcommand{\cJ}{\mathcal{J}}

\newcommand{\cP}{\mathcal{P}}
\newcommand{\cQ}{\mathcal{Q}}

\newcommand{\cX}{\mathcal{X}}
\newcommand{\cY}{\mathcal{Y}}

\begin{document}

\title{Break-Resilient Codes with Loss Tolerance}

% \title{Adversarial Break-Resilient Codes}

% \author{Canran Wang,~\IEEEmembership{Member,~IEEE},
%         Minghui Liwang~\IEEEmembership{Senior Member,~IEEE},
%         and Netanel Raviv,~\IEEEmembership{Senior Member,~IEEE}
%     \thanks{This work was supported in part by NSF under
%     Grant CNS 2223032.}
%     \thanks{Canran Wang is an independent researcher.
%     Minghui Liwang is with the Department of Control Science and Engineering, and the National Key Laboratory of Autonomous Intelligent Unmanned Systems, Tongji University, Shanghai, China.
%     Netanel Raviv is with the Department of Computer Science and Engineering, Washington University in St. Louis, St. Louis, MO 63103 USA.}
% }

\author{%
  \IEEEauthorblockN{\textbf{Canran Wang}\IEEEauthorrefmark{1},
                    \textbf{Minghui Liwang}\IEEEauthorrefmark{3},
                    and \textbf{Netanel Raviv}\IEEEauthorrefmark{1}}
  \IEEEauthorblockA{\IEEEauthorrefmark{1}%
                   Department of Computer Science and Engineering, Washington University in St. Louis,}
  \IEEEauthorblockA{\IEEEauthorrefmark{3}%
                    Department of Control Science and Engineering, Tongji University,}
  % \IEEEauthorblockN{\texttt{canran@wustl.edu}, \texttt{jsima@illinois.edu}, \texttt{netanel.raviv@wustl.edu}}
}

\maketitle
% Examples include DNA storage, where sequences might be broken during synthesis, and 3D-printing forensics, where an identifier is embedded inside a 3D printed object (e.g., a ghost gun), and then broken by an adversary in an attempt to hinder forensic investigation.
% Examples also arise in communication links, where rapid maneuvers or intermittent line-of-sight blockage can fragment a frame into disjoint bursts, and brief dropouts may erase some bursts entirely at the receiver.
\begin{abstract}
Emerging applications in manufacturing, wireless communication, and molecular data storage require robust coding schemes that remain effective under physical distortions where codewords may be arbitrarily fragmented and partially missing.
To address such challenges, we propose a new family of error-correcting codes, termed $(t,s)$-break-resilient codes ($(t,s)$-BRCs).
A $(t,s)$-BRC guarantees correct decoding of the original message even after up to~$t$ arbitrary breaks of the codeword and the complete loss of some fragments whose total length is at most~$s$.
This model unifies and generalizes previous approaches, extending break-resilient codes (which handle arbitrary fragmentation without fragment loss) and deletion codes (which correct bit losses in unknown positions without fragmentation) into a single information-theoretic framework.
We develop a theoretical foundation for $(t,s)$-BRCs, including a formal adversarial channel model, lower bounds on the necessary redundancy, and explicit code constructions that approach these bounds.
\end{abstract}

\section{Introduction}

Modern data embedding techniques increasingly operate outside traditional digital channels, in environments where data can be literally broken apart.
Consider a digital identification code embedded within a 3D-printed component for authentication or traceability: if the object is broken (intentionally or by accident) into several pieces, the codeword stored in its structure will likewise split into multiple fragments with no obvious order.
Further challenge arises by the possibility that some fragments may be lost entirely---an adversary or environmental damage could destroy certain pieces, removing those bits from observation~\cite{wang2025secure}.
Similar issues arise in DNA-based storage, where data are encoded across many short strands and are retrieved by random sampling from a DNA pool, yielding unordered reads that must be robustly reassembled~\cite{shomorony2021dna}. Wireless links can exhibit the same behavior: in UAV and satellite communications, rapid motion or intermittent blockage can split a transmitted frame into short bursts that arrive out of order, with some bursts missed, leaving only a partial, unordered view of the codeword~\cite{moon2025generalized,hashima2025next}.
Robust operation in these settings therefore calls for coding schemes that tolerate both fragmentation and missing pieces.

These examples highlight a challenging problem space on encoding information so that it remains recoverable even if the encoded data is fragmented and some pieces are entirely missing.
Traditional error-correcting codes do not readily address this combined fragmentation-and-loss problem.
On one hand, deletion codes~\cite{sima2019two,sima2020optimal,sima2023correcting,wang2023non} are designed to recover from unknown deletions within codewords, but they typically assume that the overall order of the remaining symbols is preserved.
On the other hand, recent studies have proposed novel coding schemes to handle out-of-order codeword fragments, whether the fragmentation occurs at evenly spaced intervals~\cite{sima2021coding,sima2020robust,lenz2019coding}, via a probabilistic process~\cite{shomorony2021torn,shomorony2020communicating,ravi2021capacity,ravi2024recovering}, or in an adversarial manner~\cite{bar2023adversarial,wang2024break}.
However, these approaches either provide no resilience to fragment loss, or only tolerate probabilistic losses, or limit losses to fragments whose lengths are logarithmic in the codeword length.

In this correspondence, we consider the~\emph{adversarial} version of the fragmentation-and-loss problem, in which we assume a knowledgeable adversary who is fully aware of the coding scheme and is only limited by the number of breaks and the number of lost bits.
Specifically, we introduce the~$(t,s)$-break-resilient code ($(t,s)$-BRC) to jointly handle~\emph{arbitrary} fragmentation and~\emph{arbitrary} fragment losses in a unified framework.
Our model is a natural extension of the adversarial fragmentation model introduced in~\cite{wang2024break}, which allows the adversary to break a codeword at~$t$ arbitrary locations.
On top of that, we allow the adversary to omit some subset of those fragments from the decoder whose total length does not exceed~$s$.
The remaining fragments are given to the decoder in an unordered fashion, whose task is to reconstruct the exact original message in~\emph{all cases} that respect~$t$ and~$s$.
 
We emphasize that our formulation is both adversarial and information-theoretic: it imposes no probabilistic assumptions on the break locations or the selection of lost fragments, nor does it rely on any specific physical mechanisms (such as the geometry of the fracture surface) beyond the abstract constraints defined by the parameters~$t$ and~$s$.
In particular, the knowledgeable adversary may deliberately select fragmentation and omission patterns that are maximally disruptive within the allowed limits; the worst-case assumption ensures that~$(t,s)$-BRCs are suitable for security-sensitive applications.

Moreover, our formulation is not a straightforward combination of break-resilient codes and deletion codes, since it captures the intrinsic coupling between fragmentation and deletion, as a fragment must first be separated from the codeword before it can be deleted.

\subsection{Our Contributions}
The primary contribution of this work is a theoretical framework for encoding and decoding under the combined fragmenting-and-loss adversarial model, along with corresponding explicit code constructions.
Our results demonstrate that it is possible to design codes that are both break-resilient and loss-tolerant, with redundancy~$\Theta(t\log^2n + s\log n)$.
Interestingly, our construction adopts an error locator polynomial approach similar of classical Reed--Solomon decoding\footnote{Due to space constraints, some proofs are presented in the full version~\cite{wang2025break}.}
.

\section{Problem Definition}
Our setup includes an encoder that maps an information string $\bfx \in \bi^k$ to a codeword $\bfc \in \bi^n$, with $n > k$.
For security parameters~$t$ and~$s$, an adversary may break~$\bfc$ at arbitrary~$t$ locations or less, producing at most~$t+1$ fragments.
He may then omit an arbitrary subset of these fragments whose total length is at most~$s$ bits.
The decoder is provided with the remaining fragments, which together contain at least~$n-s$ bits of the original codeword, but arrive as an unordered multiset with no positional information.

\begin{example}
    For~$\bfx=0100011100$,~$t=3$ and~$s=3$, the possible fragment multisets include
\small\begin{equation*}
    \{\{  1000,1110 \}\}, \{\{00111,00 \}\}\text{, and }\{\{ 01,11100 \}\}.
\end{equation*}\normalsize
\end{example}

A code $\cC \subseteq \bi^n$ is said to be a $(t,s)$-break-resilient code if, for every codeword $\bfc \in \cC$, the original message $\bfx$ can be uniquely recovered from any multiset of fragments produced by up to $t$ breaks and arbitrary omission of fragments totaling at most $s$ bits.
The primary metric of interest of a~$(t,s)$-BRC is its~\emph{redundancy}, i.e., the quantity~$n-\log|\cC|$, where~$|\cC|$ denotes the code size, and the~$\log$ is in base two.

We assume that the fragments are \emph{oriented}.
That is, for a given fragment~$\bff=(c_{i},c_{i+1},\ldots,c_{i+l})$ for some~$i$ and some~$l$, taken from a codeword~$\bfc=(c_1,\ldots, c_n)$, the decoder does not know the correct values of the index~$i$, but knows that the index of the leftmost bit is smaller than that of the rightmost bit.
Throughout this paper, we use standard string notation: concatenation is denoted by $\circ$;~$|\bfx|$ and~$\hw(x)$ indicates the length and Hamming weight of a string $\bfx$, respectively; and for any string $\bfx = (x_1,\ldots,x_n)$ and indices $1 \le a \le b \le n$, we write $\bfx[a:b] = (x_a,\ldots,x_b)$ and $\bfx[a:] = (x_a,\ldots,x_{|\bfx|})$.
We also write $[n] \triangleq \{1, 2, \ldots, n\}$.

\section{Bounds}
In this section, we derive a lower bound on the redundancy of a $(t,s)$-BRC $\cC$. 
We begin with the following definition.

\begin{definition}[$(t,s)$-confusable]
    Two sequences~$\bfx, \bfy \in \bi^n$ are said to be~\emph{$(t,s)$-confusable} if there exist procedures, each operating under the~$(t,s)$-constraint mentioned earlier, that transform~$\bfx$ and~$\bfy$, respectively, into an identical multiset of fragments.
\end{definition}

Clearly, a~$(t,s)$-BRC cannot contain any pair of $(t,s)$-confusable sequences as codewords.
As we will demonstrate in the following two lemmas, this observation implies that the code, or its subcodes, must have a minimum Hamming distance exceeding a certain threshold.
This property further allows us to apply the classical sphere-packing bound in the subsequent analysis to derive a lower bound on its redundancy.

\begin{lemma}\label{lemma:d-less-s}
    Let~$\cC\subseteq\bi^n$ be a $(t,s)$-BRC, and let~$\cC=\cD_1\cup\cD_2\cup\ldots,\cup,\cD_{\ceil{n/(s+1)}}$ be its partition to subcodes, where
    \small\begin{equation*}
        \cD_i = \{\bfc\in\cC \mid (s+1)\cdot i\le\hw(\bfc) < (s+1)(i+1)\}.
    \end{equation*}\normalsize
    Then, the minimum Hamming distance of each~$\cD_i$
    \small\begin{equation}
        d_\text{min}(\cD_i)\geq \ceil{\frac{t}{2}}.
    \end{equation}\normalsize
    % is at least~$\ceil{\frac{t}{2}}$.
\end{lemma}
\begin{proof}
    Assume for contradiction that there exist distinct~$\bfx,\bfy\in\cD_i$ such that~$d=d_H(\bfx,\bfy)<\ceil{\frac{t}{2}}$ for some~$i\in[0,\ceil{n/(s+1)}]$.
    As we will demonstrate next, this assumption implies that~$\bfx$ and~$\bfy$ are~$(t,s)$-confusable.
    
    Let~$j_1,\ldots,j_{d}\in[n]$ be the indices where~$\bfx$ and~$\bfy$ differ.
    Then,~$\bfx$ and~$\bfy$ can be written as
    \small\begin{align*}
        \bfx&=\bfc_1\circ x_{j_1}\circ\bfc_2\circ\ldots\circ \bfc_d\circ x_{j_d}\circ\bfc_{d+1},\\
        \bfy&=\bfc_1\circ y_{j_1}\circ\bfc_2\circ\ldots\circ \bfc_d\circ y_{j_d}\circ\bfc_{d+1}
    \end{align*}\normalsize
    for some (potentially empty) sequences~$\bfc_1,\ldots,\bfc_{d+1}$, in which~$x_{j_l}\ne y_{j_l}$ for every~$l\in[d]$.
    Consider an adversary that breaks~$\bfx$ and~$\bfy$ to the immediate left and the immediate right of every entry~$j_l$.
    Since~$2\cdot d<2\cdot\ceil{\frac{t}{2}}$ --- the right side of the inequality equals to~$t$ when~$t$ is even and~$t+1$ when~$t$ is odd --- the operation is permitted under the~$t$-constraint, and results in the following two multisets:
    \small\begin{equation}\label{eq:xy-multiset}
        \begin{split}
                \cX &= \{\{\bfc_1,\ldots,\bfc_{d+1},x_{j_1},\ldots,x_{j_d}\}\},\\
        \cY &= \{\{\bfc_1,\ldots,\bfc_{d+1},y_{j_1},\ldots,y_{j_d}\}\}.
        \end{split}
    \end{equation}\normalsize
    Without loss of generality, let~$\hw(\bfx)\geq\hw(\bfy)$.
    Define
    \small\begin{equation}\label{eq:define-sp}
        \begin{split}
        s'&\triangleq\hw(x_{j_1},\ldots,x_{j_d})-\hw(y_{j_1},\ldots,y_{j_d})\\
        &\overset{\eqref{eq:xy-multiset}}{=}\hw(\bfx)-\hw(\bfy)\leq s,
        \end{split}
    \end{equation}\normalsize
    where the last inequality follows from the fact that~$\bfx,\bfy\in\cD_i$ and~$\cD_i$ contains all codewords of~$\cC$ with Hamming weight ranging from~$(s+1)\cdot i$ to~$(s+1)(i+1)-1$.
    
    As a result, the adversary can create two identical multisets by omitting~$s'$~$1$'s from~$\{x_{j_1},\ldots,x_{j_d}\}$, and~$s'$~$0$'s from~$\{y_{j_1},\ldots,y_{j_d}\}$, i.e.,
    \small\begin{equation*}
    \cX\setminus\underbrace{\{\{1,\ldots,1\}\}}_{s'\mbox{ times}} = \cY\setminus\underbrace{\{\{0,\ldots,0\}\}}_{s'\mbox{ times}}.   \end{equation*}\normalsize
    Note that the above removal operations are permitted due to the following reasons.
    \begin{enumerate}
        \item There are sufficiently many~$1$'s in~$\cX$ since\small\begin{equation*}\hw(x_{j_1},\ldots,x_{j_d})\overset{\eqref{eq:define-sp}}{=}s'+\hw(x_{j_1},\ldots,y_{j_d})\geq s'.\end{equation*}\normalsize
        \item There are sufficiently many~$0$'s in~$\cY$ since
        \small\begin{align*}
            d-\hw(y_{j_1},\ldots,y_{j_d})\overset{\eqref{eq:define-sp}}{=}d-\hw(x_{j_1},\ldots,y_{j_d})+s'\geq s'.
        \end{align*}\normalsize
        \item The~$s$-constraint is satisfied since~$s'\leq s$.
    \end{enumerate}
    Thus,~$\bfx$ and~$\bfy$ are~$(t,s)$-confusable, a contradiction.
\end{proof}

The above lemma holds for all values of~$t$ and~$s$.
In the regime where~$t = o(n)$ and~$s = \delta n$ (i.e., when~$t$ is asymptotically small relative to~$n$ and~$s$ constitutes a fixed fraction of~$n$), the following lemma establishes a bound on the minimum Hamming distance of~$\cC$. 

\begin{lemma}\label{lemma:d-more-s}
    Let~$\cC\subseteq\bi^n$ be a $(t,s)$-BRC with~$s\ge\ceil{\frac{t}{2}}$.
    Then, its minimum Hamming distance
    \small\begin{equation*}
     d_\text{min}(\cC)>\frac{\mu\cdot n - 2}{\mu+2},~\mbox{where}~\mu =\frac{t}{n-s}.
    \end{equation*}\normalsize
\end{lemma}
\begin{proof}
Assume for contradiction that there exist~$\bfx,\bfy\in\cC$ such that~$d_H(\bfx,\bfy)\le\theta$.
We will show that this assumption implies that~$\bfx$ and~$\bfy$ are~$(t,s)$-confusable.

Let~$j_1,\ldots,j_{d}\in[n]$ be the indices where~$\bfx,\bfy$ differ, where
\small\begin{equation}\label{eq:min-distance}
    d=d_H(\bfx,\bfy)\leq\frac{\mu\cdot n - 2}{\mu+2} = \frac{t\cdot n - 2\cdot (n-s)}{t+2\cdot(n-s)}.
\end{equation}\normalsize
Then,~$\bfx$ and~$\bfy$ can be written as
    \small\begin{align*}
        \bfx&=\bfc_1\circ x_{j_1}\circ\bfc_2\circ\ldots\circ \bfc_d\circ x_{j_d}\circ\bfc_{d+1},\\
        \bfy&=\bfc_1\circ y_{j_1}\circ\bfc_2\circ\ldots\circ \bfc_d\circ y_{j_d}\circ\bfc_{d+1}
    \end{align*}\normalsize
for some (potentially empty) sequences~$\bfc_1,\ldots,\bfc_{d+1}$, in which~$x_{j_l}\ne y_{j_l}$ for every~$l\in[d]$.

Define the following~\emph{knockout} string operation: select a substring, break at its immediate left and right, and then omit the substring.
Consider the following \emph{baseline} procedure: knocking out substrings of length one at every position~$j_l$ in~$\bfx$ and~$\bfy$, which results in two identical multisets~$\{\{\bfc_1,\ldots,\bfc_{d+1}\}\}$; this procedures requires~$2 \cdot d \le t$ and~$d \le s$.

If~$2 \cdot d \le t$, the condition~$d \le s$ is automatically satisfied given our assumption that~$s \ge \ceil{t/2}$.
Hence, the baseline procedure complies with both the~$t$- and~$s$-constraints, implying that~$\bfx$ and~$\bfy$ are~$(t,s)$-confusable—a contradiction.

\begin{figure}[t]
	\centering
	\begin{subfigure}{0.49\textwidth}
		\includegraphics[width=1\textwidth]{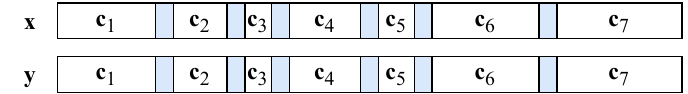}
		\caption{}
	\end{subfigure}
    \hfill
	\begin{subfigure}{0.49\textwidth}
		\includegraphics[width=1\textwidth]{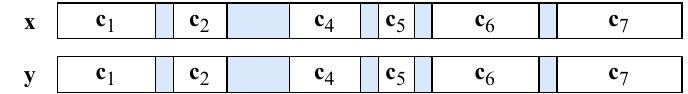}
		\caption{}
	\end{subfigure}\\
    \centering
	\begin{subfigure}{0.49\textwidth}
		\includegraphics[width=1\textwidth]{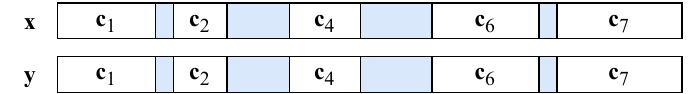}
		\caption{}
	\end{subfigure}
        \hfill
	\begin{subfigure}{0.49\textwidth}
		\includegraphics[width=1\textwidth]{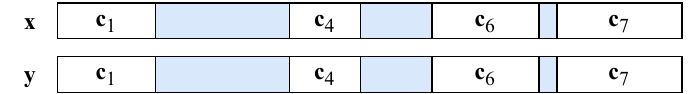}
		\caption{}
	\end{subfigure}
 	\caption{Illustration of the merge operations for the case~$t = 6$ and sufficiently large~$s$, where the colored substrings indicate the substrings to be knocked out.
    (a) Initially, the baseline procedure involves knocking out all different bits between~$\bfx$ and~$\bfy$, which requires~$12 > 6$ breaks and thus exceeds the~$t$-constraint.
    (b) The result of applying a merge operation to~$\bfc_3$, the \textit{shortest} among the~$\bfc_i$'s, reducing the number of required breaks to~$10$ at the price of increasing the number of omitted bits by~$|\bfc_3|$.
    (c) The result of applying a merge operation to~$\bfc_5$, the \textit{second} shortest among the~$\bfc_i$'s, reducing the number of required breaks to~$8$ at the price of increasing the number of omitted bits by~$|\bfc_5|$.
    (d) The result of applying a merge operation to~$\bfc_2$, the \textit{third} shortest among the~$\bfc_i$'s, reducing the number of required breaks to~$6$ at the price of increasing the number of omitted bits by~$|\bfc_2|$.}
    \label{fig:merge}
\end{figure}

If~$2 \cdot d > t$, the adversary cannot employ the baseline procedure described earlier due to the~$t$-constraint.
Yet, the adversary may attempt to update the procedure to reduce its number of required breaks, at the cost of increasing the number of required omissions.
To this end, we introduce the following \emph{merge} operation:
Given a substring~$\bfc_l$, instead of knocking out the two neighboring substrings separately, the adversary removes all three (the two neighbors and~$\bfc_l$ itself) as a single contiguous substring.
This reduces the required number of breaks by~$2$, at the cost of omitting an additional~$\abs{\bfc_l}$ bits\footnote{Note that the merge operation is not a string operation, but rather an operation on the adversary’s procedure which~\emph{merges} two konckout string operations into one.}.

The adversary repeats this merge operation, prioritizing the shortest~$\bfc_l$'s, until the required number of breaks meets the~$t$-constraint.
An illustration of this process is provided in Figure~\ref{fig:merge}.
Let~$\bfc_{i_1}, \bfc_{i_2}, \ldots, \bfc_{i_w}$ denote the~$w$ additional substrings omitted due to the merge operations.
The number of required breaks in the updated procedure is then~$2d - 2w$.
Since the merging continues until the~$t$-constraint is met, it follows that~$w = d -\floor{\frac{t}{2}}$.

Recall that the baseline procedure omits~$d$ bits. Hence, the number of omitted bits in the updated procedure is

\small\begin{align*}
    &d + (\abs{\bfc_{i_1}}+\cdots +\abs{\bfc_{i_w}})=d +w\cdot\operatorname{avg}(\abs{\bfc_{i_1}},\ldots,\abs{\bfc_{i_w}})\\
    \overset{(a)}{\le} & d +w\cdot\operatorname{avg}(\abs{\bfc_{1}},\ldots,\abs{\bfc_{d+1}})=d+w\cdot \frac{n-d}{d+1} \\
    =&d+(d -\floor{\frac{t}{2}})\cdot \frac{n-d}{d+1}
    \leq d+(d -{\frac{t}{2}}+1)\cdot \frac{n-d}{d+1}\\
    =&n-\frac{t}{2}\cdot\frac{n-d}{d+1},
\end{align*}\normalsize
where~(a) follows from the fact that~$\bfc_{k_1},\ldots,\bfc_{k_w}$ are the shortest among~$\{\bfc_1,\bfc_2,\ldots,\bfc_{d+1}\}$.
By Equation~\eqref{eq:min-distance}, we have

\small\begin{equation*}
    \begin{split}
        n-\frac{t}{2}\cdot\frac{n-d}{d+1}&\leq n-\frac{t}{2}\cdot\frac{1}{d+1}\cdot \left(n-\frac{t\cdot n - 2\cdot (n-s)}{t+2\cdot(n-s)}\right)\\
        & \leq n-\frac{t}{2}\cdot \frac{t+2\cdot(n-s)}{t\cdot(n+1)}
        \cdot\frac{2\cdot(n-s)\cdot(n+1)}{t+2\cdot(n-s)}\\
        &=n-(n-s)=s,
    \end{split}
\end{equation*}\normalsize
which shows that the number of omitted bits in the updated procedure is no more than~$s$, satisfying the~$s$-constraint.
Hence,~$\bfx$ and~$\bfy$ are~$(t, s)$-confusable, a contradiction.
\end{proof}

Based on the two preceding lemmas, we proceed to derive a lower bound on the redundancy of any~$(t, s)$-BRC.

\begin{theorem}\label{theorem:bound-on-redundancy}
    A~$(t,s)$-BRC~$\cC$ satisfies
    \small\begin{equation*}
        n-\log|\cC|\ge \Omega(n \cdot \mu\log\textstyle\frac{1}{\mu}),~\mbox{where~$\mu=\frac{t}{n-s}$.}
    \end{equation*}\normalsize
\end{theorem}
\begin{proof}
    In the case where~$0\leq s < \ceil{\frac{t}{2}}$, let~$\cD_i$ be the set of all codewords of~$\cC$ of Hamming weight ranging from~$(s+1)\cdot i$ to~$(s+1)(i+1)-1$ for~$i\in[\ceil{\frac{n}{s+1}}]$, as defined in Lemma~\ref{lemma:d-less-s}.
    Let~$\cD_{i_\text{max}}$ be the largest set among the~$\cD_i$'s.
    We have that
    \small\begin{align*}        \log|\cC|&=\log(\textstyle\sum_{i=1}^{\ceil{\frac{n}{s+1}}}|\cD_i|)\le \log(\ceil{\frac{n}{s+1}}\cdot |\cD_{i_\text{max}}|)\\
    &\le\log(\frac{n}{s}\cdot |\cD_{i_\text{max}}|) \leq \log n - \log s+\log|\cD_{i_\text{max}}|.
    \end{align*}\normalsize
    Furthermore, it follows from Lemma~\ref{lemma:d-less-s} that~$\cD_{i_\text{max}}$ is of minimum Hamming distance at least~$\ceil{\frac{t}{2}}$.
    Therefore, by the sphere-packing bound we have that 
    \small\begin{align*}
        |\cD_{i_{\text{max}}}|\le \frac{2^n}{\sum_{j=0}^{t'}\binom{n}{j}}, \text{ where }t'\triangleq \bigg\lfloor\frac{\ceil{\frac{t}{2}}-1}{2}\bigg\rfloor\approx\frac{t}{4}.
    \end{align*}\normalsize
    Hence, it follows that
    \small\begin{align}\label{eq:r-less-s}
    \textstyle    n-\log|\cC|&\ge n-\log n + \log s-\log|\cD_{i_\text{max}}| \nonumber\\
        &\ge n-\log n +\log s-n+\log\left( \textstyle\sum_{j=0}^{t'}\binom{n}{j} \right)\nonumber\\
        &\ge \log s+ \log\textstyle\binom{n}{t'}-\log n\nonumber\\
        &\overset{(\star)}{\ge}  \log((n/t')^{t'})-\log n\nonumber\\
        &=t'\log n-t'\log(t')-\log n=\Omega(t\log \tfrac{n}{t}),
    \end{align}\normalsize
    where~$(\star)$ is a easily provable. 
    
    In the case where~$s \geq \ceil{\frac{t}{2}}$, we again apply the the sphere-packing bound, with~$\theta$ as defined in Lemma~\ref{lemma:d-more-s}:
    \small\begin{align}\label{eq:r-more-s}
  n-\log|\cC|&\geq n - \log \frac{2^n}{\sum_{j=0}^\theta\binom{n}{j}}=\textstyle\log \left(\sum_{j=0}^\theta\binom{n}{j}\right)\nonumber\\
  &\geq \textstyle\log\binom{n}{\theta}\geq \log(n/\theta)^\theta =\theta \cdot\log (n\cdot \frac{\mu +2 }{\mu\cdot n -2}) \nonumber\nonumber\\ 
  &\geq \textstyle\frac{\mu\cdot n -2}{\mu+2}\cdot \log (1+\frac{2}{\mu})
    \geq \frac{\mu\cdot n -2}{3} \cdot  \log\left( \frac{2}{\mu}\right)\nonumber\\
    &= \textstyle\Omega( n\cdot \mu\log\frac{1}{\mu}).
\end{align}\normalsize
Note that when~$0\leq s<\ceil{\frac{t}{2}}$,~\eqref{eq:r-more-s} becomes
    \small\begin{equation*}
        \begin{split}
        \textstyle\Omega(n\cdot \frac{t}{n-s}\log\frac{n-s}{t}) = \textstyle\Omega( \frac{n}{n-s}\cdot t \cdot \log (\frac{n}{t}-\frac{s}{t}))=\textstyle\Omega(t\log\frac{n}{t}),
        \end{split}
    \end{equation*}\normalsize
which matches~\eqref{eq:r-less-s}.
Together, the redundancy of a~$(t,s)$-BRC~$\cC$ satisfies~$n-\log|\cC|\ge \Omega(n \cdot \mu\log\frac{1}{\mu})$.
\end{proof}

\section{Break-Resilient Codes with Loss Tolerance}\label{section:BRC}

In this section, we introduce~$(t,s)$-BRC, a family of codes that simultaneously achieve~$t$-break resiliency and~$s$-loss tolerance.
Our construction is analogous to classical Reed--Solomon decoding:
 define a locator polynomial~$\Lambda$, compute syndromes from the surviving fragments, and recover the missing information as its roots.
Our construction incurs $O(t\log^2 n+s\log n)$ redundancy bits.
We present the encoding and decoding procedures and prove correctness.

\begin{algorithm}[t]\caption{\textsf{encode} ($(t,s)$-BRC Encoding)}\label{alg:Encode}
%\begin{multicols}{2}
\begin{algorithmic}[1]
\Statex \textbf{Input:}~{A valid binary string~$\bfz\in\bi^m$.}
\Statex \textbf{Output:}~{A~$(t,s)$-BRC codeword $\bfc\in\bi^n$}.
    \State Let~$
        \cP \triangleq \{\textsf{bin}^{-1}(\bfz[i,i+M -1]): i\in [0,m-M] \}$\label{line:cp}
    \State Let~\texttt{CKSUMS} be an array s.t.\label{line:define-PS}
    \small\begin{equation*}
        \textstyle\texttt{CKSUMS}[j]=\sum_{p\in\cP}{(\beta_j-p)^{-1}}.
    \end{equation*}\normalsize%for~$j\in[L]$.
    \State $\bfr\gets \textsf{bin}(\texttt{CKSUMS}[1])\circ\bfm_1\circ\cdots\circ \textsf{bin}(\texttt{CKSUMS}[L])\circ\bfm_{L}$\label{line:define-r}
    \State \textbf{Output}~$\bfc\gets\bfz\circ\bfm_0\circ\bfr$.\label{line:define-codeword}
\end{algorithmic}
%\end{multicols}
\end{algorithm}

\begin{algorithm*}\caption{\textsf{decode} ($(t,s)$-BRC Decoding)}\label{alg:Decode}
\begin{algorithmic}[1]
\Statex \textbf{Input:} A multiset \texttt{F} of fragments produced by an~$(t,s)$-constrained adversary from some~$\bfc\in \cC$.
\Statex \textbf{Output:} String $\bfz\in\bi^m$ for which Alg.~\ref{alg:Encode} yields~$\bfc$.
\Statex \textcolor{red}{$\triangleright~\textsc{classification of fragments}$}
\If{there exists $\bff\in\texttt{F}$ and index~$i$ such that~$\bfm_0=\bff[i:i+c\log m-1]$}
    \State Remove~$\bff$, and add $\bff[0:i-1]$ and $\bff[i:]$ to \texttt{F}.\label{line:transpoint}
\EndIf
\State Let~$\texttt{F}_\text{redundancy}$ be the set of fragments that contains~$\bfm_l$ for some~$l\in[0,L]$ (fragments in the redundancy region).\label{line:R-FRAGMENTS}
\State Let~$\texttt{F}_\text{information}$ be the remaining fragments whose length are at least~$3 M+1$ bits (fragments in the information region).\label{line:Z-FRAGMENTS}

\Statex \textcolor{red}{$\triangleright$~\textsc{construction of subset~$\cQ\subset\cP$} and array~$\texttt{CKSUMS}_\text{obs}$}

\State Initialize an empty set~$\cQ \gets \emptyset$ and an empty array $\texttt{CKSUMS}_\text{obs}$ of length~$L$.
\ForAll{fragments~$\bff$ in $\texttt{F}_\text{information}$}\label{line:fetch-start}
    \For{$i=0$ \textbf{to} $|\bff|-c\log m$}
        \State $\sigma \gets \bff[i,\, i+c\log m-1]$, $q \gets \textsf{bin}^{-1}(\sigma)$, and~$\cQ \gets \cQ \cup \{q\}$\label{line:fetch-end}
        % \State $\cQ \gets \cQ \cup \{q\}$\label{line:fetch-end}
    \EndFor
\EndFor

% \Statex \textcolor{red}{$\triangleright$~\textsc{construction of array~$\texttt{CKSUMS}_\text{obs}$}}
% \State Let $\texttt{CKSUMS}_\text{obs}$ be an array s.t.,~$\texttt{CKSUMS}_\text{obs}[l]=\texttt{EMPTY}$  for~$l\in[L]$\label{line:fetch-r-start}
% \State Initialize an empty array $\texttt{CKSUMS}_\text{obs}$ of length~$L$.\label{line:fetch-r-start}
\ForAll{fragments~$\bff$ in $\texttt{F}_\text{redundancy}$, and if there exists index~$i$ such that~$\bfm_l=\bff[i:i+M]$ for some~$l\in[L]$}\label{line:fetch-r-start}
    % \If{there exists index~$i$ such that~$\bfm_l=\bff[i:i+M]$ for some~$l\in[L]$}
        \State $\texttt{CKSUMS}_\text{obs}[l]\gets\bff[i-M:i-1] $~\textbf{if}~$i\geq M$
        \State $\texttt{CKSUMS}_\text{obs}[l+1]\gets\bff[i+M+1:i+2\cdot M] $~\textbf{if}~$i+2\cdot M < |\bff|$\label{line:fetch-r-end}
    % \EndIf
\EndFor
\Statex \textcolor{red}{$\triangleright$~\textsc{decode from locator polynomial}}
\State Let~$\Lambda(z)=z^r+b_{r-1}z^{r-1}+\cdots+b_1z+b_0$ with~$r=|\cP|-|\cQ|$ unknown variables~$b_0,\ldots,b_{r-1}$.
\State For each~$j\in\cJ\triangleq\{\texttt{CKSUMS}_\text{obs}[j]\neq\texttt{EMPTY}\}$,~$S_j\gets \texttt{CKSUMS}_\text{obs}[j]-\sum_{q\in\cQ}({\beta_j-q})^{-1}$.
\State Solve for~$b_0,\ldots,b_{r-1}$ using system of linear equations~$\{S_j\cdot \Lambda(\beta_j)- {\Lambda'(\beta_j)}=0\}_{j\in\cJ}$.
\State $\cP\gets\cQ\cup\{\text{roots of } \Lambda(z) \}$, add map~$\cP$ back to~$\bfz$\label{line:output-z}
\State \textbf{Output}~$\bfz$.
\end{algorithmic}
\end{algorithm*}

\subsection{Encoding}\label{section:encoding}

The encoding of our~$(t,s)$-BRC scheme begins by selecting a uniformly random string~$\bfz\in\bi^{m}$ to serve as the information word, and rejecting (and resampling) it unless $\bfz$ is \emph{valid}, as defined below.

Fix a (not necessarily integer) constant $c>2$ such that~$c\log m$ is an integer.
Let
\small\begin{equation*}
    M=c\log m+1,~\mbox{and}~ L=t\cdot(3M)+s,
\end{equation*}\normalsize
let $\textsf{bin}:\bbF_{2^{M}}\to \bi^{M}$
be a bijection, and let $\textsf{bin}^{-1}$ be its inverse.
Fix $L+1$ distinct elements $\{\bfm_0,\ldots,\bfm_L\}\subset \cC_{\text{MU}}$, where~$\cC_{\text{MU}}$ is a binary mutually-uncorrelated code\footnote{A mutually-uncorrelated code is a set of codewords such that no proper prefix of any codeword equals a suffix of any (possibly the same) codeword.} of block length~$M+1$.
We associate each marker $\bfm_l$ with a unique field element
\small\begin{equation*}
    \beta_l \triangleq \textsf{bin}^{-1}(\bfm_l[0:M-1])\in \bbF_{2^M}, \qquad l=0,\ldots,L,
\end{equation*}\normalsize
i.e., the binary representation of~$\beta_l$ is the first~$M$ bits of~$\bfm_l$.
\begin{definition}\label{def:valid}
A binary string~$\bfz\in\bi^m$ is~\emph{valid} if:
\begin{enumerate}[start=1,label={(\Roman*)}]
    \item \label{item:distinct}
    Every two substrings of length~$M-1$ of~$\bfz$ are distinct.
    \item \label{item:no-markers}
    $\bfz$ does not contain any of the markers $\bfm_0,\ldots,\bfm_L$.
\end{enumerate}
\end{definition}

These criteria allow the encoder to exploit $\bfz$ to generate redundancy bits and thus produce a $(t,s)$-break-resilient codeword.
Moreover, choosing MU-codewords as markers of length $M+1$ ensures that the decoder can unambiguously distinguish markers from all other substrings of the codeword.

\begin{remark}
    We comment that choosing the information word at random is mere convenience, which in practice implies that our encoding algorithm applies only for a fraction of all possible information words.
    It will be clear in Section~\ref{section:redundancy} that this fraction goes polynomially fast to~$1$ as~$m$ goes to infinity, and thus induces at most one extra bit of redundancy.
    Mapping a string of length~$k=m-1$ to a valid string of length~$m$ efficiently should be possible using techniques such as repeat-free codes~\cite{elishco2021repeat}, and is left as future work.
\end{remark}

Algorithm~\ref{alg:Encode} presents the encoding procedure of our scheme, which takes a valid string $\bfz\in\bi^m$ as input and outputs a $(t,s)$-BRC codeword.
Let~$\cP\subset\bbF_{2^M}$ denote the set\footnote{$\cP$ is a set (rather than a multiset) by Property~\ref{item:distinct} of valid strings $\bfz$.} of all substrings of length~$M=c\log m+1$ of~$\bfz$ interpreted as elements in~$\bbF_{2^M}$ (Alg.~\ref{alg:Encode}, line~\ref{line:cp}).
The following lemma allows us to cast decoding as a set reconciliation problem.
\begin{lemma}\label{lemma:bijective}
    The mapping from~$\bfz$ to~$\cP$ is injective. %\cc{Not bijective,}
\end{lemma}
\begin{proof}
    We show that $\cP$ uniquely determines $\bfz$ by giving a deterministic reconstruction procedure.
    Pick any $p\in\cP$, set $\bfy\gets \textsf{bin}(p)$, and remove $p$ from $\cP$.
    If there exists $p_\ell\in\cP$ s.t.
    \small\begin{equation*}
        \textsf{bin}(p_\ell)[1:c\log m]=\bfy[0:c\log m -1],
    \end{equation*}\normalsize
     then update $\bfy\gets \textsf{bin}(p_\ell)[0]\circ \bfy$ and remove $p_\ell$ from $\cP$.
    If there exists $p_r\in\cP$ s.t.,
    \small\begin{equation*}
        \textsf{bin}(p_r)[0:c\log m-1]=\bfy[|\bfy|-c\log m:|\bfy|-1],
    \end{equation*}\normalsize
    then update $\bfy\gets \bfy\circ \textsf{bin}(p_r)[c\log m]$ and remove $p_r$ from $\cP$.    
    Repeat the above extensions until~$\cP$ is empty.
    By Property~\ref{item:distinct}, a matching predecessor~$p_\ell$ or successor~$p_r$, if exists, is \emph{unique}, and hence the reconstruction process is unique as well.
\end{proof}

As stated in line~\ref{line:define-PS}, the encoder defines an array~\texttt{CKSUMS} to store the~$L$ checksums of the $p_i$'s: for each $j\in[L]$,
\small\begin{equation}
    \textstyle\texttt{CKSUMS}[j]\triangleq\sum_{p\in\cP}({\beta_j-p})^{-1}.
\end{equation}\normalsize
Due to Property~\ref{item:no-markers}, all entries of~\texttt{CKSUMS} are well-defined since~$p \ne \beta_j$ for every~$j\in[L]$ and every~$p\in\cP$.
The encoder then constructs the redundancy by interleaving checksums with the markers.
As seen in line~\ref{line:define-r}, the redundancy is
\small\begin{equation}
\bfr \triangleq \textsf{bin}(\texttt{CKSUMS}[1])\circ\bfm_1\circ\cdots\circ \textsf{bin}(\texttt{CKSUMS}[L])\circ\bfm_{L}.
\end{equation}\normalsize
Finally, the encoder forms the codeword by concatenating $\bfz$ and $\bfr$, with the marker $\bfm_0$ inserted in between (line~\ref{line:define-codeword}), i.e.,
\small\begin{equation}
\bfc \triangleq \bfz \circ \bfm_0 \circ \bfr.
\end{equation}\normalsize

\subsection{Decoding}

The goal of decoding is to reconstruct the set $\cP$ from fragmented and partially missing codeword pieces, and then use it to uniquely recover $\bfz$ (Lemma~\ref{lemma:bijective}).
Specifically, the decoder first identifies the information region (i.e., the $\bfz$ portion of $\bfc$) and the redundancy region.
From fragments in the information region, it extracts a subset $\cQ\subseteq \cP$. Together with the checksums contained in the redundancy region, the decoder is able to find the missing elements and recover~$\cP$.

\addtolength{\topmargin}{0.1in} %Fuck EDAS

The classification of fragments is demonstrated in Alg.~\ref{alg:Decode}, line~\ref{line:transpoint}--\ref{line:Z-FRAGMENTS}.
Note that only fragments of length at least~$3 M+1$ bits are included in~$\texttt{F}_\text{information}$ to ensure they come entirely from the information region. Indeed, any fragment of length at least~$3 M+1$ that intersects the redundancy region must contain a marker, and would have been classified into~$\texttt{F}_\text{redundancy}$.

With~$\texttt{F}_\text{information}$, the subset $\cQ\subseteq\cP\subset \bbF_{2^{M}}$ is constructed as follows.
For each fragment $\bff\in \texttt{F}_\textsf{information}$ and each $i\in[0,|\bff|-M]$, let $q\in\bbF_{2^{M}}$ be the field element whose binary representation equals the length-$M$ substring $\bff[i, i+M-1]$.
The decoder then inserts $q$ into $\cQ$ (Alg.~\ref{alg:Decode}, line~\ref{line:fetch-start}--\ref{line:fetch-end}).
Given~$\texttt{F}_\text{redundancy}$, the decoder fetches checksums using the markers (line~\ref{line:fetch-r-start}--\ref{line:fetch-r-end}).

Let~$\cJ$ be the set of all indices~$j$ such that~$\texttt{CKSUMS}_\text{obs}[j]$ is non-empty.
Then,
% Then, with proof presented in the full version~\cite{wang2025break},
\begin{theorem}\label{theorem:qjp}
    $|\cQ| + |\cJ|\geq |\cP|$.
\end{theorem}
\begin{proof}
    Recall that a knockout operation is specified by a starting index $i$ and a length $b$, and omit the substring $\bfc[i,\, i+b-1]$ from $\bfc$. The adversary performs at most~$t$ knockout operations and omits at most $s$ bits in total.

    Partition the knockout operations into those whose deleted bits lie in the information region and those whose deleted bits lie in the redundancy region.
    If a knockout operation crosses both regions, treat it as two knockout operations that act on the information region and the redundancy region, respectively: namely, split the deleted interval into its intersection with the information region and its intersection with the redundancy region.
    The newly introduced break does not affect the number of lost redundancy bits, nor does it affects the number of missing elements in~$\cP$.

    Consider a knockout operation of length~$b$ starting at~$i$ within the information region, i.e., which removes the interval~$[i,i+b-1]$.
    An element $p_\ell\in\cP$ corresponding to the substring
    $\bfz[\ell,\,\ell+M-1]$ for some~$\ell$ would be missing as a result of this knockout operation only if $[\ell,\,\ell+M-1]$ intersects~$[i,\,i+b-1]$.
    The set of such $\ell$'s is contained in $[i-M+1,\, i+b-1]$, hence at most~$b+M-1$ elements of $\cP$ are affected by this knockout.
    Assuming there exist~$t_\text{I}$ knockouts of lengths~$b_1,\ldots,b_{t_\text{I}}$ and denoting~$\sum_ib_i=s_\text{I}$, it follows that
    the number of missing elements from $\cP$ is at most~$t_\text{I}(M-1)+s_\text{I}$.

    Note that the decoder ignores fragments shorter than~$3M+1$ to avoid confusion between the information region and the redundancy region (Alg~\ref{alg:Decode}, line~\ref{line:Z-FRAGMENTS}).
    These fragments reside between two knocked out fragments, and affects~$\cP$ as if they were knocked out.
    Hence, one such fragment introduces at most~$2M+1$ extra missing elements (corresponds to all of its substrings of length $M+1$) of ignored fragments) in~$\cP$. 
    Therefore,
    \small\begin{equation*}
        |\cQ|\ \ge\ |\cP|-(t_\text{I}(M-1+2M+1)+s_\text{I} )=|\cP|-(t_\text{I}\cdot 3M+s_\text{I}).
    \end{equation*}\normalsize

    A knockout of length $b$ within the redundancy region destroys at most~$1+\left\lceil \frac{b}{2M+1}\right\rceil$ checksums.
    Hence, the total number of destroyed checksums is at most
    \small\begin{equation*}
        \sum_{\ell=1}^{t_\text{R}} \left(1+\ceil{\frac{b_\ell}{2M+1}}\right)\leq t_\text{R}+\sum_{\ell=1}^{t_\text{R}}\left(\frac{b_\ell}{2M+1}+1\right)=2t_\text{R}+\frac{s_\text{R}}{2M+1},
    \end{equation*}\normalsize
    Therefore,
    \small\begin{equation*}
        |\cJ| \geq\ L-\left(2t_\text{R}+\frac{s_\text{R}}{2M+1}\right).
    \end{equation*}\normalsize
    Putting the two bounds together yields
    \small\begin{align*}
    |\cQ|+|\cJ|
    &\ge |\cP|-(t_{\sf I}(3M)+s_{\sf I})
        + L-\left(2t_{\sf R}+\frac{s_{\sf R}}{2M+1}\right) \\
    &= |\cP| +  (L-t\cdot 3M-s)\geq \cP.\hspace{2.5cm}\qedhere
    \end{align*}\normalsize
\end{proof}

Let~$r=|\cP|-|\cQ|$ denote the number of elements missing from~$\cP$.
For some (unknown)~$b_0,...,b_{r-1}\in \bbF_{2^{M}}$, define locator polynomial
\small\begin{equation}
  \textstyle  \Lambda(z)=\prod_{p\in\cP\setminus\cQ}(z-p)=z^r+b_{r-1}z^{r-1}+\cdots+b_1z+b_0.
\end{equation}\normalsize
Consider the formal derivative~$\Lambda'(z)$ of~$\Lambda(z)$, defined via the operator which applies~$(x^i)'=(i \mod 2)\cdot x^{i-1}$ for every integer~$i \ge 0$, and extends linearly.
Then, the following is an immediate corollary of standard derivation rules:
\small\begin{equation}\label{eq:derivative}
    \textstyle {\Lambda'(z)}/{\Lambda(z)}=\sum_{p\in\cP\setminus\cQ}(z-p)^{-1}.
\end{equation}\normalsize
For every~$j\in\cJ$, the decoder computes a syndrome
\small\begin{align*}
    S_j= \texttt{CKSUMS}_\text{obs}[j]-\sum_{q\in\cQ}\frac{1}{\beta_j-q}=\prod_{d\in\cP\setminus\cQ}\frac{1}{\beta_j-d}  \overset{\eqref{eq:derivative}}{=}\frac{\Lambda'(\beta_j)}{\Lambda(\beta_j)}.
\end{align*}\normalsize
Thus, for every~$j\in\cJ$, we have an equation of~$r$ variables:
\small\begin{equation}
   S_j\cdot \Lambda(\beta_j)- {\Lambda'(\beta_j)}=0.
\end{equation}\normalsize
By Theorem~\ref{theorem:qjp}, the decoder obtains a system of~$|\cJ|$ linear equations with~$r=|\cP|-|\cQ|\leq |\cJ|$ variables.
By solving the system, it obtains the polynomial~$\Lambda(z)$, whose roots are the elements missing from~$\cP$.
Finally, once $\cP$ is recovered, the decoder uniquely determines $\bfz$ (Lemma~\ref{lemma:bijective}), completing the proof of the following theorem.
\begin{theorem}
    Alg.~\ref{alg:Decode} outputs $\bfz$ for which Alg.~\ref{alg:Encode} yields~$\bfc$.
\end{theorem}

\section{Redundancy Analysis}\label{section:redundancy}
The next theorem bounds the success probability
of choosing a valid binary string~$\bfz$~(Definition~\ref{def:valid}); it is inspired by ideas from~\cite[Theorem~4.4]{cheng2018deterministic} and analogous to~\cite[Theorem~6]{wang2024break}, and the proof is given in~\cite{wang2025break}.
\begin{theorem}\label{theorem:valid-probability}
    A uniformly random string $\bfz\in\bi^m$ is valid (Definition~\ref{def:valid}) with probability~$1-1/\operatorname{poly}(m)$.
\end{theorem}
\begin{proof}
    Property~\ref{item:distinct} is essentially the $B$-distinct property in~\cite{cheng2018deterministic}, and we follow their argument.
    With~$i<j$, consider any two length-$(c\log m)$ substrings
    \begin{equation*}
        \bfz_i \triangleq \bfz[i:i+c\log m-1]~\mbox{and}~\bfz_j \triangleq \bfz[j:j+c\log m-1].
    \end{equation*}
     In the case where~$j-i\geq c\log m$, 
    \small
    \begin{align*}
        \Pr\left(\bfz_i=\bfz_j\right)&\textstyle =\sum_{\bfr\in\bi^{c\log m}}\Pr\left(\bfz_j=\bfr\mid \bfz_i=\bfr\right)\cdot\Pr\left(\bfz_i=\bfr\right)\\
     &= {2^{-c\log m}}.
    \end{align*}
    \normalsize
    In the case where~$j-i< c\log m$, the intervals~$\bfz_i$ and~$\bfz_j$ overlap.
    Thus, letting~$\bfr\triangleq\bfz[i:j-1]$,
    \small\begin{equation*}
        \bfz_j=\left(\bfr\circ\bfr\circ\bfr\circ\cdots\right)[0:c\log m-1],
    \end{equation*}\normalsize
    i.e., the~$(c\log m)$-prefix of a concatenation of an infinite series of~$\bfr$.
    Hence,~$\bfz_j$ is completely determined by~$\bfr$, and as a result,
    \small\begin{equation*}
    \begin{split}
      \textstyle  \Pr(\bfz_i=\bfz_j)=\sum_{\bfr\in\bi^{j-i}}\Pr(\bfz_j=\bfr\circ \bfr\circ\cdots)[0:c\log m-1]\\ \hfill\mid \bfz[i:j-1]=\bfr)\cdot \Pr(\bfz[i:j-1]=\bfr)={2^{-c\log m}}.\\
    \end{split}
    \end{equation*}\normalsize
    
    By applying the union bound, the probability that such~$i,j$ exist is at most~$\frac{m^2}{2^{c\log m}} =m^{-c+2}$.
 
    For Property~\ref{item:no-markers}, note that the probability that an interval of length~$c\log m$ of~$\bfz$ matches one of the markers is~$(L+1)m^{-c}$.
    By applying the union bound, the probability that such an interval exists is at most
    \begin{equation*}
        \frac{(t+1)(m-c\log m+1)}{m^c}<\frac{(L+1)m}{m^c}<m^{-c+2}.    
    \end{equation*}

    To conclude, a uniformly random~$\bfz\in\bi^m$ does not satisfy either of the two properties with probability~$1/\operatorname{poly}(m)$ each.
    Therefore,~$\bfz$ is valid, i.e., satisfies two properties, with probability at least~$1-1/\text{poly(m)}$ by the union bound. 
\end{proof}

With Theorem~\ref{theorem:valid-probability}, we can formally provide the redundancy of our scheme in the following corollary.
\begin{corollary}
    The~$(t,s)$-BRC code has redundancy of
    \small\begin{equation*}
        O(((2\cdot M+1)\cdot L)=O(t\log^2 n+s\log n).
    \end{equation*}\normalsize
\end{corollary}

\clearpage
\printbibliography

\end{document}